\documentclass[11pt]{article}

\usepackage{hyperref}

\usepackage{booktabs} 

\usepackage[ruled]{algorithm2e} 

\SetAlFnt{\small} \SetAlCapFnt{\small} \SetAlCapNameFnt{\small}
\SetAlCapHSkip{0pt} \IncMargin{-\parindent}

\usepackage{fullpage}
\newcommand{\E}{\mathbb{E}}
\usepackage{amsmath,amsfonts,amssymb,amsthm,enumitem}
\usepackage{color}
\newcommand{\AB}[1]{#1}
\newcommand{\notyet}[1]{}
\newcommand{\greedy}{{\tt Greedy}}
\newcommand{\batch}{{\tt Batch}}

\newtheorem{theorem}{Theorem}[section]

\newtheorem{lemma}[theorem]{Lemma}
\newtheorem{claim}[theorem]{Claim}
\newtheorem{corollary}[theorem]{Corollary}
\newtheorem{definition}[theorem]{Definition}

\begin{document}
\title{Kidney exchange and endless paths: On the optimal use of an altruistic donor
}
\author{Avrim Blum\thanks{Toyota Technological Institute at Chicago.  This work was supported in part by the National Science Foundation under grant CCF-1733556.}
\and Yishay Mansour\thanks{Blavatnik School of Computer Science,
Tel-Aviv University and Google Research. Supported in part by a
grant from the Israel Science Foundation, a grant from the United
States-Israel Binational Science Foundation (BSF), and the Israeli
Centers of Research Excellence (I-CORE) program (Center No. 4/11).}}
\maketitle

\begin{abstract}
We consider a \AB{well-studied} online random graph model for kidney exchange, where nodes representing patient-donor pairs arrive over time, and the probability of a directed edge is $p$.  We assume existence of a single altruistic donor, who serves as a start node in this graph for a directed path of donations.  The algorithmic problem is to select which donations to perform, and when, to minimize the amount of time that patients must wait 
before receiving a kidney.

\AB{We advance our understanding of this setting by (1) providing efficient (in fact, linear-time) algorithms with optimal $O(1/p)$ expected waiting time, (2) showing that some of these algorithms in fact provide guarantees to all patients of $O(1/p)$ waiting time {\em with high probability}, (3) simplifying previous analysis of this problem, and (4) extending results to the case of multiple altruistic donors. \notyet{and (5) extending results to a model with multiple types}}


\end{abstract}


\section{Introduction}
Altruistic donors have proven to be very powerful in practice in
kidney exchange, with a single donor enabling a long sequence of
matches.   In fact, such sequences have a name: a {\em Never Ending
Altruistic Donor (NEAD) chain} \cite{Rees09}.\footnote{See also
{\url{
http://www.nationalkidneycenter.org/treatment-options/transplant/a-chain-of-hope/nead-chain/}.}}
The idea is that the altruistic donor donates to a compatible
patient $A$ who has already joined the kidney exchange with her
willing but incompatible donor $B$, and in return $B$ agrees to
``pay it forward'' by serving as a donor to some existing or future
compatible patient $C$ who has entered the system with her willing
but incompatible donor $D$, and so on.  One reason this can be so
powerful is that unlike cyclic exchanges, these donations do not
have to be simultaneous.\footnote{In a cyclic exchange, where, say,
donor $D$ donates to patient $A$ and donor $B$ donates to patient
$C$, the operations need to be simultaneous since if one donor were
to back out after the other donor has donated, the un-transplanted
patient would have lost their donor.}

In this work, we consider a \AB{well-studied} online random graph model \AB{\cite{AshlagiJM13,AndersonAGK15,AndersonAGK17,AshlagiBJM2016,ashlagi2019matching}} in which
nodes (patient-donor pairs) arrive over time, and between any two
nodes $u,v$ there is a directed edge with probability $p$ (with
probability $p$, the donor for $u$ is compatible with the patient
for $v$).  
\AB{This model is of particular interest when $p$ is small, which corresponds to the important case of highly-sensitized patients.}
Our goal is to minimize the average waiting time of
patients until they get a kidney, as well as to provide per-patient high-probability bounds on their waiting time.

There are two natural extremes for using an altruistic donor.  One
is greedy longest-waiting-time-first: whenever a patient arrives who is compatible with the
altruistic donor, immediately donate the kidney, making the
patient's associated donor become the new altruistic donor.  There now may be multiple patients who this donor can donate to, and in this algorithm we always choose the patient who has been waiting the longest.  In a sense, this is the most ``fair'' and natural algorithm.
 A second natural extreme is to wait until enough
patient-donor pairs have arrived so that there is a Hamiltonian path
visiting all nodes in the graph, and then completely clear the queue
using this Hamiltonian path. 
\AB{Both algorithms yield an
$\Theta(\frac{1}{p}\log \frac{1}{p})$ expected waiting time per
patient.  However, this is not optimal, and  \cite{ashlagi2019matching} give an alternative, computationally-inefficient algorithm that achieves an optimal bound of $\Theta(\frac{1}{p})$.  In this work, we give two computationally-efficient (in fact, linear time) procedures that achieve the optimal $\Theta(\frac{1}{p})$
expected waiting time per patient, one of which also provides each patient a guarantee of waiting time at most $O(\frac{1}{p})$ with high probability.  We also give a somewhat simpler correctness argument.}  For the first algorithm, the idea is to
wait a bit before matching (unlike the greedy algorithm) in order to
have more options for routing a long path, but not to the extreme of
requiring the path cover all the existing nodes in the graph. 
\AB{It then uses Depth First
Search (DFS) on the graph to discover a long (though perhaps not the longest) path,} and therefore is
linear time to implement. We bound not only the expected waiting
time but also provide high probability bounds.  We then use our analysis of 
\AB{this algorithm to analyze a second}
 algorithm that does not require waiting.  This algorithm
runs the greedy longest-waiting-time-first algorithm when the number of patients waiting is small, but then switches over to one of the other algorithms (using a DFS-based method to select a path rather than always choosing the patient who has been waiting the longest) in order to more quickly reduce the queue.  We show the combined algorithm also 
enjoys an optimal
$\Theta(\frac{1}{p})$ expected waiting time for each patient.

In addition, we
also consider the case of multiple altruistic donors, and show that
if there are only $O(1/p)$ altruistic donors then the $\Omega(1/p)$ lower
bound still applies, while for $\Omega((1/p)\log(1/p))$ donors even a naive greedy has an $O(1)$ expected waiting time. (However, the $O(1)$ might be misleading, since conditioned on the patient not being matched
immediately, then the expected waiting time is $\Omega(1/p)$.)

\notyet{Finally, we consider a model where donor-patient pairs belong to different types, and there is a matrix describing which types can have edges to which other types, and prove results in this setting as well.}

We note that in real-life kidney exchanges, there is some chance
donors will back out (or become ill or otherwise be unable to
donate).  We ignore this effect here because it obscures the
distinctions between different algorithms.  For instance, if each
donor backs out with probability $q$, then no algorithm can possibly
hope to construct a chain of expected length more than $1/q$.

\subsection{Related work}

There has been substantial work analyzing kidney exchange in static
random graph models.  Questions studied include whether it is
possible to match most patients in the system, and to what extent
long chains in addition to short cycles are needed
\cite{AshlagiR14,AGRR12}.   Additionally, researchers have
considered questions such as motivating hospitals to join and fully
participate in the exchange under such models
\cite{AshlagiR11,ToulisP11} and how to match in the presence of
failures \cite{DPT13}.

\cite{Unver10} was the first to consider kidney exchange in {\em
dynamic} random graph models in which nodes arrive one at a time.
This work considers the dense-graph case, focusing on blood-type
incompatibility rather than highly-sensitized patients as we do
here.

The work of \cite{AshlagiJM13} considers dynamic kidney exchange in
the model we consider here, namely highly-sensitized patients (small
$p$) and assuming all pairs are blood-type compatible (so for every
pair $u,v$ there is an edge from $u$ to $v$ with probability $p$).
One of their
main results is that allowing for a
chain in addition to cycles of size $2$ or $3$ increases the \AB{total number of matches} linearly in the number of arriving nodes.
%
\cite{AkbarpourSG14} consider a dynamic model in which nodes both
arrive and depart over time, and examine different pairwise matching
algorithms in this model.

\AB{The question of queue size (expected waiting time for patients to receive a kidney) in the setting of highly sensitized patients is considered in \cite{AndersonAGK15,AndersonAGK17,AshlagiBJM2016,ashlagi2019matching}.}
\cite{AndersonAGK15} examine queue size for cycles
rather than paths.  Their conclusion is that for cycles, greedy
matching is optimal: they
show that the greedy algorithm has an average waiting time of
$O(1/p^2)$ for cycles of size $2$ and $O(1/p^{3/2})$ for cycles of
size $2$ or $3$, which is best possible. In a follow-up work,
\cite{AndersonAGK17}  show that for a chain (which starts with an
altruistic donor) the greedy algorithm, if it selects the longest
path in the graph, can guarantee an expected waiting time of
$O(1/p)$.
A clear caveat of such an approach is that computing the longest path is NP-complete.
In contrast, our efficient algorithms run in linear time.
%
%
\AB{Finally, \cite{AshlagiBJM2016,ashlagi2019matching} consider a dynamic model with both
easy-to-match and hard-to-match patients, and consider both cycles
and paths, analyzing expected waiting time.  A particularly relevant result shown in \cite{ashlagi2019matching} for our setting is}
that the greedy algorithm for paths has average waiting time
of $\Theta((1/p)\log(1/p))$ when all patients are hard to match.

There is of course a substantial body of work in general on analysis
of random graphs; see, e.g.,
\cite{bollobas2001random,frieze2015introduction}.
We use the ideas from \cite{Krivelevich16} to compute long paths in
random graphs using DFS algorithms.


\section{Model}

In our model, the basic unit is a pair consisting of a patient and a
willing but incompatible donor, which we model as a node in a
directed graph $G(V,E)$. The set of nodes of $V$ are these
patient-donor pairs except for one special node which represents the
single altruistic donor, which we call the start node.
%
A directed edge between two nodes indicates that the donor of the
first pair (node) is compatible with the patient of the second pair
(node).

The process of matching donors with patients reduces to finding a
directed path starting at the altruistic donor (start node). Again,
the interpretation is that each directed edge $u\rightarrow v$
represents a donation from the donor at $u$ to the patient at $v$. The
number of edges in the path represents the number of donations.
Essentially, our goal is to maximize the length of the path,
maximizing the number of patients that receive a donation.

We consider an {\em online (dynamic) model}, where there is a stream of
nodes (patient and donor pairs) that arrive one per time step and
our goal is to minimize the expected time a patient waits until she
is matched to a compatible donor.

%
%
We assume that we start at time $t=0$ with the start node $v_0$
(altruistic donor). At each integer time $t>0$ one node $v_t$ arrives. For
each existing node $v_\tau$, for $\tau < t$, we select with
probability $p$ an incoming edge and with probability $p$ an
outgoing edge, independently. I.e., with probability $p$ we have
$v_\tau \rightarrow v_t$, and also with probability $p$ we have $v_t
\rightarrow v_\tau$, where all the events are independent.

In time $t$, the algorithm may extend the directed path (which began
originally at the start node $v_0$, the altruistic donor) by one or
more edges if such an extension exists in the graph.  This is viewed
as servicing or matching those nodes on the directed path.

At each time $t$ we have a node $v_t^e$ which is the end of the
current path, and we call it {\em the end of the path}, and any
future extension has to start with it.
The nodes which are still not on the directed path are called {\em
waiting nodes}. We refer to the {\em queue size}, $q_t$, at time
$t$, as the number of waiting nodes at time $t$. The waiting time
$w_t$ of a node $v_t$ is the time between its arrival, $t$, and the
time it is added to the path $a_t$ (namely, the time until the
patient in the patient-donor pair is serviced). Formally,
$w_t=a_t-t$.

We assume that extending the path, by any extension, is done
instantaneously, and we ignore that time. We also assume that a node
exists until it is added to the path (i.e., serviced). Namely, patients do not depart until they receive a kidney.

Our main discussion is on {\em when} and {\em how} to extend
the path. Unlike some online models, an arriving node does not have
to be serviced immediately, even if it can be.

We assume that the process continues for $T$ time steps, but most of
our results will be {\em independent} of this parameter. For
nodes $v_t$ that are not serviced by time $T$ we assume that
$a_t=T$, just for simplicity of the presentation.

\subsection{Evaluation criteria}

Our main evaluation criteria is the expected waiting time of a node,
or alternatively, the expected queue size.
Note that the sum of the waiting times of nodes (i.e., patients)
over time is the same as the sum of the queue sizes over
time.\footnote{The sum of the waiting times is $\sum_t w_t = \sum_t
a_t -t =\sum_t \sum_\tau I[t\leq \tau\leq a_t]$, where $I[\cdot]$ is
the indicator function. The sum of the queue lengths is $\sum_\tau
q_\tau = \sum_\tau\sum_t I[t\leq \tau\leq a_t]$.}
%

\begin{definition}
The {\em average waiting time (queue size)} is
$\frac{1}{T}\sum_{t=1}^T q_t = \frac{1}{T}\sum_{t=1}^T w_t$. The
{\em expected waiting time (queue size)} is
$\E[\frac{1}{T}\sum_{t=1}^T q_t] = \E[\frac{1}{T}\sum_{t=1}^T w_t]$,
where the expectation is with respect to the random edges and any
randomization of the path selection algorithm.
\end{definition}

We will be also interested in deriving high probability guarantees
on the waiting time of a node.

\begin{definition}
We say that node $v_t$ has with probability $1-\delta$ a waiting
time of at most $\omega$ if $\Pr[w_t > \omega ]\leq \delta$.
\end{definition}

We are also interested in giving node specific guarantees, starting
at an arbitrary time, conditioned on the history. The goal is to
show that the system does not discriminate against any patient, and that
history has a limited effect on the waiting time.

\begin{definition}
The {\em expected additional waiting time} of node $v_t$ given a
history $h_\tau$ until time $\tau > t$ is $0$ if we have that
$a_t\leq \tau$, and otherwise it is $\E[a_t - \tau| h_\tau]$.
\end{definition}

\subsection{Multiple altruistic donors}

%


We also consider the case that there are multiple altruistic donors,
denoted by $R$. In this case each of the $R$ altruistic donors has a
separate directed path. (Clearly, the paths are node disjoint.) At
each time $t$, the algorithm decides which of the paths to extend
and how to extend them. (The algorithm may decide to extend multiple
paths at the same time.)

%
%
%

\section{Preliminaries: Random walks}

In many places in our analysis we will need to analyze sequences of random
variables which are generated through a random walk. At a high
level, the sequence of non-negative random variables will have the
property that if their value is above a certain threshold, we are
guaranteed that in expectation their value will decrease.
Intuitively, this implies that their expected value cannot be much
larger then the threshold. While this holds, under some assumptions
that do hold in our setting, it does require some analysis that we
perform in the Appendix.

Two remarks are in order. First, we believe that our derivation is most
likely implicitly known, but unfortunately we were not able to locate
any reference. For this reason we added the derivation in the
Appendix. Second, we remark that one cannot use the Azuma inequality in our
setting, since the decrease is ``unbounded'', while the Azuma
inequality requires that the maximum change is bounded.

%
%



We now do the precise formalization and derivation.
Let $Q_t$ be a sequence of non-negative random variables where
initially $Q_1=0$. At time $t+1$ either $Q_{t+1}=Q_t+1$ or,
$Q_{t+1}=Q_t-Z_t$ and $Z_t \in [0,Q_t]$. The main property that we
assume about the sequence is that when $Q_t\geq M+1$ then with
probability at least $\rho$ we have $Z_t\geq K$. In addition, $\rho
K =1+\beta $ where $\beta >0$, which means that for $|Q_t|\geq M+1$,
the expected change in $Q_t$, which is at most $1-\rho K$, is
negative. We call such a random walk a $(M,K,\rho,\beta)$ random
walk. (Actually, since $\rho K=1+\beta$ one parameter is redundant,
but it will be more convenient to have all four parameters.)

For a $(M,K,\rho,\beta)$ random walk, when $Q_t \geq M+1$ we have
that $E[Q_{t+1}|Q_t]\leq Q_t - \beta$. Intuitively we like to claim
that this implies that $E[Q_t]\leq M+O(K)$, however, this requires
some care.

In Appendix~\ref{app:random-walk} we show the following theorem.

\begin{theorem}
\label{thm:rand-walk}
 Let $Q_t$ be a $(M,K,\rho,\beta)$ random walk,
where $\beta \leq 3/5$. Then $E[Q_t]\leq M+K(1+\beta)/\beta$. In
addition, with probability $1-\delta$ we have $Q_t\leq
M+\frac{K(1+\beta)}{\beta}\ln \frac{2}{\delta}$.
\end{theorem}



\section{Lower bound for any algorithm}

We start by showing that for any algorithm the expected waiting time
has to be $\Omega(1/p)$.

\begin{theorem}
\label{thm:lower-bound-any}
For any algorithm the expected waiting time is at least
$0.5/p$.
\end{theorem}

\begin{proof}
Fix any time $t$ and consider $v_t$ the node arriving at time $t$.
We need $v_t$ to have at least one incoming edge to have it served.
The expected number of coin flips until we have an incoming edge is
$1/p$. Therefore, we have $\E[w_t | q_t] \geq 1/p - q_t$, since we
immediately do $q_t$ coin flips, for the $q_t$ waiting nodes, and
the right-hand-side corresponds to a lower bound in which we
consider an edge to the $i$th waiting node as giving $w_t=i-q_t$
rather than giving $w_t=0$.

Taking the expectation with respect to the history up to time $t$,
we have $\E[w_t] \geq 1/p-\E[q_t]$. Averaging over all time steps we
have $\frac{1}{T} \sum_t (\E[w_t] +\E[q_t]) \geq \frac{1}{p}$. Since
the expected average waiting time and queue size are identical, we
have that the expected average waiting time is at least
$\frac{1/2}{p}$.
\end{proof}

%

\section{Greedy algorithm}

We now concentrate on the simple greedy longest-waiting-time-first algorithm (\greedy).  \greedy\ extends
$v^e_t$, the end of the path at time $t$, the first opportunity it
has, and in the event of multiple options always chooses the patient who has been waiting the longest.  Note that after it completes an extension, the new
end of the path, $v^e_{t+1}$, does not have any directed edges to
waiting nodes. We \AB{begin with} matching upper and lower bounds on the
expected waiting time of the greedy algorithm. \AB{These bounds are shown also in \cite{ashlagi2019matching} but we prove them here through a different argument that helps to set up our general methodology.}

\begin{theorem}
\greedy\ has an expected waiting time of
$\Theta(\frac{1}{p}\log\frac{1}{p})$, for $p<1/2$. In addition, with
probability $1-\delta$ the waiting time is at most
$O(\frac{1}{p}\log\frac{1}{\delta p})$,
\end{theorem}

\begin{proof}
We start by showing that the waiting time of \greedy\ is at least
$\Omega(\frac{1}{p}\log\frac{1}{p})$.

At time $t$, with probability $1-p$, node $v_t$ does not have an incoming edge from
 the end of the directed path $v^e_t$. In this case, the queue size
grows by $1$.  Otherwise,
with probability $p$, node $v_t$ has an incoming edge from $v^e_t$, and
$v_t$ and a path of length $path(v_t)$ extends the current path from
$v^e_t$. The change in the queue size $q_t$ is,
\[
\E[q_{t+1}|q_t]=q_t +(1-p) - p \E[path(v_t)]\;.
\]
We need to upper bound $path(v_t)$. Note that while generating the
path, we have a probability of at least $(1-p)^{q_t}$ of terminating
the path since we reached a node with zero outgoing
degree.\footnote{We are sampling the out edges of the node only when
we add it to the path.  This is legal because nodes are never
revisited and because we choose which neighbor to visit next based only on its time of arrival and not based on which edges it has.  Other selection rules may behave differently, 
as \cite{AndersonAGK17} show.} This implies that
\[
\E[q_{t+1}|q_t] \geq q_t +(1-p) - p \frac{1}{(1-p)^{q_t}}\;.
\]
Now, for $q_t < \lambda_\epsilon$, where $\lambda_\epsilon =(1/p)\ln
((1-p-\epsilon)/p)$, we have
\[
\E[q_{t+1}|q_t] \geq q_t +\epsilon\;.
\]

We can now partition the time into intervals, where $q_t <
\lambda_\epsilon$, or singletons where $q_t \geq \lambda_\epsilon$.
In the intervals, we have an increase of $q_t$ bounded by $1$
(deterministically) and at least $\epsilon$ (in expectation). Assume
we start an interval with a value $q_t= 0$ (this will be the worst
case). The expected length of the interval would be at most
$\lambda_\epsilon/\epsilon$ and the sum of the queue lengths would
be at least $\lambda^2_\epsilon/2$. With probability at least $1/2$
the length of the interval is at most $2\lambda_\epsilon/\epsilon$.
This implies that the average queue size in the interval is at least
\[
\frac{1}{2}\frac{\lambda^2_\epsilon/2}{2\lambda_\epsilon/\epsilon} =
\frac{\epsilon\lambda_\epsilon}{8}\;.
\]
For $\epsilon=1/2$ we get an expected lower bound of
$(1/16)\lambda_{0.5} =(1/(16p))\ln ((0.5-p)/p)$.

We now analyze the upper bound using $(M,K,\rho,\beta)$ random walk
and Theorem~\ref{thm:rand-walk}. Notice that given that $v^e_t$, the
end of the path, has an outgoing edge (probability $p$ at each time
$t$), the probability of extending by a path of length at least
$4/p$ is at least $(1-(1-p)^{q_t})^{4/p}$. For $q_t >(1/p)\log
(4/p)$ this is at least $1/e$. This implies that we have a
$(M,K,\rho,\beta)$ random walk for $M=(1/p)\log (4/p)$, $K=4/p$,
$\rho=p/e$ and $\beta=4/e-1<0.5$. From Theorem~\ref{thm:rand-walk}
we have the desired upper bound.
\end{proof}

\section{The CLEAR-ALL algorithm}

The {\tt CLEAR-ALL} algorithm waits until it can extend the current
path and serve all the existing nodes, i.e., using a Hamiltonian
path. This implies that we partition the time to phases, where the
algorithm {\tt CLEAR-ALL} serves all the waiting nodes. This implies
that each phase starts with an empty queue!

\begin{theorem}
The {\tt CLEAR-ALL} algorithm has an expected waiting time of
$\Theta(\frac{1}{p}\log\frac{1}{p})$, for $p<1/2$.
\end{theorem}

\begin{proof}
%
The algorithm {\tt CLEAR-ALL} waits until the waiting nodes have a
Hamiltonian path connecting all of them. From graph theory we know that for
an Erdos-Renyi graph $G(n,p)$ if
$p=\frac{\log(n)+\log\log(n)+O(\log(1/\delta))}{n}$ then with
probability $1-\delta$ we have a Hamiltonian path (see,
\cite{KomlosS83,Bollabas84,Frieze1988}).

This implies that when we have $
n_\delta=\frac{O(\log(1/p)+\log(1/\delta))}{p}$ nodes waiting, with
probability $1-\delta$ we have a Hamiltonian path. This implies that
the expected number of arrivals before we have a Hamiltonian path is
$\Theta(\frac{1}{p}\log (\frac{1}{p}))$. Each time we have a Hamiltonian path in the
graph of the waiting nodes, we extend the current path from the end
of the path $v^e_t$, using the Hamiltonian path and completely empty
the queue of waiting nodes.
\end{proof}

\section{Batch algorithm}

We now present the \batch\ algorithm, which we show achieves waiting time only $\Theta(\frac{1}{p})$.
The idea behind the algorithm is to wait for some time
and aggregate arrivals, and then in one time step to compute and add
a long path. The benefit, compared to \greedy, is that
we can plan better to find a longer path. The challenge is that now the graph on the nodes left over is no longer random, because the path is determined algorithmically based on structural properties of the nodes.  In contrast, one of the key features of \greedy\ is that because it selects which outgoing edge to take based solely on the arrival time of the incident node, the graph on unvisited nodes remains uniform random.   In fact, we will use this property later to show that if desired, we can replace the waiting step in \batch\ with runs of a greedy algorithm, producing a hybrid algorithm that always makes a match when one is available and yet still achieves expected waiting time $\Theta(\frac{1}{p})$.

The \batch\ algorithm has a parameter $c>0$ and works in phases.
At the start of each phase we wait $c/p$ time, for $c/p$ incoming
nodes to arrive. We then run a procedure {\tt PATH} that extends
the current path. Then we start a new phase. Different
implementation use different procedures {\tt PATH}.

The following is a description of an implementation of {\tt PATH}
which we call {\tt FAIR-PATH}. The procedure {\tt FAIR-PATH} works
as follows. Let $Q$ be the set of waiting nodes at the start of
the phase and let $V_{fp}$ be the set of $c/p$ nodes that arrived
since the start of the phase. We build a graph
$G_{fp}(V_{fp},E_{fp})$ where $V_{fp}$ are the $c/p$ new arriving
nodes. For each old node $v\in Q$, if there are new nodes
$u_1,u_2\in V_{fp}$, such that there are edges $u_1\rightarrow v$
and $v\rightarrow u_2$, then we pick a random ingoing edge to $v$,
say from $u_1$, and a random outgoing edge from $v$, say to $u_2$,
and add an edge from $u_1\rightarrow u_2$ to $E_{fp}$ and label it
by $v$. Namely, for each $v\in Q$, let $IN(v)=\{u_1\in V_{fp}:
u_1\rightarrow v\}$ and $OUT(v)=\{u_2\in V_{fp}: v\rightarrow
u_2\}$. If $IN(v)\neq \emptyset$ and $OUT(v)\neq \emptyset$ then we
select a random $u_1\in IN(v)$ and a random $u_2\in OUT(v)$ and add
the edge $u_1\rightarrow u_2$ to $E_{fp}$ and label it $v$. This
defines the edges $E_{fp}$, and if there are multiple parallel
edges, we select one such edge at random.

Let $v^e$ be the end of the path at the end of the previous phase.
We add $v^e$ to $V_{fp}$ and add its edges to $E_{fp}$, namely,
$OUT(v^e)=\{u\in V_{fp}: v^e\rightarrow u \}$.
If $OUT(v^e)$ is empty, {\tt FAIR-PATH} returns an empty extension.


We run the algorithm DFS-LP\footnote{DFS-LP runs a DFS algorithm,
and returns the longest path it observes during its run. See
Appendix \ref{app:long} for more discussions on the topic.} from
$v^e$, and let $path(v^e)$ be the path that it returns. We extend
the current directed path using $path(v^e)$ by adding for each edge
the vertex which is its label, i.e., the nodes that caused the
insertion of that edge.

Note that the extension path, $path(v^e)$, alternates between nodes
that arrive during the last phase, i.e., nodes from $V_{fp}$, and
nodes that arrive in previous phases, i.e., nodes in $Q$. If
$path(v^e)$ has $\ell$ edges in $G_{fp}$ then we are extending by
$2\ell-1$ nodes, where $\ell$ are from the recent phase, and
$\ell-1$ are from previous phases.

Let $|Q|$ be the number of nodes that remain from previous phases.
We would like to consider nodes from $Q$ that have at least one
incoming and one outgoing degree from $V_{fp}$. The expected number
of such nodes is at least $|Q|(1-2(1-p)^{c/p})\approx|Q|(1-2e^{-c})$
and for $c\geq 10$ with high probability it is at least
$0.9Q\triangleq m$.

We now show a simple property of a random graph where the number of
edges is fixed, but the actual edges are selected uniformly at
random. Specifically, we consider now a random graph with $n=c/p$
nodes and $m$ edges, where the edges are selected at random with
replacements, so there might be multiplicities. We show that if the
number of edges is large enough, any two disjoint subsets of size
$k$ will share an edge.

\begin{lemma}
\label{lemma:random-m}
 A random graph with $n$ nodes and $m\geq
(n^2/k \log (n/(k\delta))$ random edges, with probability
$1-\delta$, for any two disjoint sets for size $k$ there is an edge
joining them.
\end{lemma}

\begin{proof}
For the proof, we do a union bound over all pairs of disjoint
subsets of size $k$. For a fixed disjoint sets $S_1$ and $S_2$ of
size $k$, the probability that a given edge will select to connect
them is $k^2/n^2$. Therefore, the probability that there are two
such sets which do not share an edge is bounded by,
\[
{n \choose k}{{n-k} \choose k} (1-\frac{k^2}{n^2})^m \leq
(\frac{en}{k})^{2k} e^{-k^2 m /n^2} \leq \delta
\]
where the last inequality uses the assumption that $m\geq (n^2/k)
\log (n/(k\delta)$.
\end{proof}

\begin{theorem}
\label{thm:batch-algorithm}
For $p<0.04$ and $c>10$, the expected waiting time in the Batch
algorithm with parameter $c$ is $O(c/p)$ and with probability
$1-\delta$ it is at most $O((c/p)\log(1/\delta))$.
\end{theorem}

\begin{proof}
Let $Q_t$ be the number of waiting nodes at the end of phase $t$. We
would like to show that the $Q_t$ forms a $(M,K,\rho,\beta)$ random
walk. However, the increase of $Q_t$ can be $c/p$ (rather than $1$
in the $(M,K,\rho,\beta)$ random walk). For this reason we  scale
each $c/p$ nodes as a ``one unit'', and show the bound for the
$(M,K,\rho,\beta)$ random walk. At the end we multiply by $c/p$ to
get the correct bound.

Let $n=c/p$, $k=n/10$, and $M=(n^2/k) \log (n/(k\delta)=(10c/p) \log
(10/\delta)$. Once we scale down by $c/p$ and set $\delta=0.1$ we
have $M=10\log(100)$. We like to compute the probability of decrease
and its magnitude. We need the magnitude to be at least $c/p$ to
have a net decrease (which is $1$ after the scaling).

First we show that if there are many waiting nodes, then with high
probability we have many edges in $G_{fp}$. Assume that $|Q_t|\geq
M+1$ and $m=0.9|Q_t|$. For any $v\in Q_t$, the probability that
$IN(v)=\emptyset$ is $e^{-c}$ and similarly $OUT(v)=\emptyset$ is
$e^{-c}$. This implies that with probability at least $1-2e^{-c}$
there is an edge labeled by $v$. The probability that we have a
duplicate edge is $2|Q_t|/(c/p)^2$, so the expected number of edges
is at least $|Q_t| (1-p^2/c^2-2e^{-c})$. For $p<0.04$ and $c\geq 10$
we have that $|E_{fp}|< m$ with probability at most $2e^{-10}$.

By Lemma~\ref{lemma:random-m}, with probability $1-\delta=0.9$,
between any two subsets of size $k$ there is an edge.
By Corollary \ref{cor:start-nodes}, there exists a path of length
$2(n-2k)= 1.6c/p$ nodes, for all but a subset $S$ of at most $k$
nodes, as a start node. The probability that $OUT(v^e)\subset S$ is
$(1-(k/n))^{1/p}=0.9^{1/p}<0.1$. Therefore, with probability
$(1-0.9^{1/p})(1-\delta)(1-2e^{-10})>0.8$ the procedure {\tt
FAIR-PATH} will extend by $2(n-2k)= 1.6c/p$ nodes. This implies that
we have $\rho=0.8$ and $K=1.6$ (after scaling down by $c/p$). We
have that $\beta=0.28$.

By Theorem~\ref{thm:rand-walk}, for such a $(M,K,\rho,\beta)$ random
walk, we have that the expected value is at most $100\log(100)+
10=O(1)$ and with probability $1-\delta$ it is at most
$O(\log(1/\delta))$. Scaling back by $c/p$ derives the theorem.
\end{proof}

\section{Not a short path}

A clear drawback of the greedy algorithm is that in many cases it
generates rather short paths to be added. The Not A Short Path
(NASP) algorithm will overcome this weakness by requiring that the
length of the path that we add is ``not short''. Specifically, the
algorithm will have a parameter $c>0$ and it will add only paths of
length at least $\theta=c/p$. This will clearly overcome the issue
of adding short paths. The challenge is that now the duration of a
phase (the time between two consecutive extensions of the path) is a
random variable. It is worthwhile to compare the NASP algorithm to
the {\tt batch} algorithm. While in the {\tt batch} algorithm the
duration of a phase is fixed and the length of the extension is a
random variable, in NASP the duration of a phase is a random
variable and the length of the extension has a fixed lower bound (we
allow to add longer paths).


The algorithm NASP works in phases. In each phase, as in the {\tt
batch} algorithm, the goal is to build an extension to the path
built from both new and old nodes. The main difference is that a
phase does not have a pre-specified number of new nodes (unlike the
{\tt batch} algorithm which had exactly $c/p$ new node arrivals).
The phase ends when we find an extension path of length at least
$\theta=c/p$ nodes. Therefore in NASP the duration of a phase is a
random variable while the length of the path is always at least
$\theta=c/p$. In contrast, in the {\tt batch} algorithm, the size of
the phase is fixed to be $c/p$ while the length of the extension
path is a random variable.

The analysis bounds the expected duration of a phase, i.e., the
expected time it takes until we have such a path.
Let $Q_t$ be the set of waiting nodes at the start of the phase $t$.
Our analysis will have two cases depending on the number of waiting
nodes nodes $|Q_t|$ at the start of the phase $t$. If the number of
waiting nodes at the start of the phase is small, we show that the
expected duration of the phase is not too large. We will not claim
much about the outcome of such a phase, just that it ends in
expected $O(c/p)$ time. If the number of waiting nodes at the start
of the phase is large, we show that with reasonable probability the
number of waiting nodes will decrease (compared to the start of the
phase). We start by considering the duration of a phase when the
number of waiting nodes is arbitrary (actually, the worse case would
be no waiting nodes).

\begin{lemma}
\label{lemma:Gnp-path}
Let $c>100$.  After $1.2c/p$ new nodes, with probability at least
$1-e^{-30}$ we have a path of length $\theta=c/p$.
\end{lemma}

\begin{proof}
Note that we make no assumption about $|Q_t|$, the number of waiting
nodes at the start of the phase. For the proof we consider only the
new arriving nodes in the phase (implicitly assuming that
$|Q_t|=0$). After $n=1.2c/p$ we have a $G(n,p)$ random graph. The
probability that there exists two subsets of size $k=0.1c/p$ nodes
that do not share an edge is
\[
{n \choose k}{{n-k} \choose k} (1-p)^{k^2} \leq ((12e)^{2}
e^{-0.1c})^k \leq e^{-30}
\]
Therefore, with probability $1-e^{-30}$, we have that every two
subsets of size $k=0.1c/p$ share an edge. By
Lemma~\ref{lemma:long-path} This implies that the graph has a path
of length at least $c/p$.
\end{proof}

\begin{corollary}
\label{cor:NASP-phase-duration}
For $\theta= c/p$, the expected duration of a phase is at most
$1.21c/p$.
\end{corollary}

\begin{proof}
By Lemma~\ref{lemma:Gnp-path} we have that after $1.2c/p$ new nodes
we have a path of length $c/p$ with probability $1-e^{-30}>0.999$.
This implies that the expected time is at most $1.21c/p$.
\end{proof}

In the above case we consider only the new nodes (implicitly assumed
$Q_t=\emptyset$). Not surprisingly, the number of waiting nodes is
likely to increase in such a case. The main benefit of NASP is that
in the case that there are many waiting nodes at the start of the
phase ($|Q_t|$ is large) then we expect that the number of waiting
nodes will decrease. Similar to the {\tt batch} algorithm ,we can
show,

\begin{claim}
\label{claim:long-path}
If $|Q_t|\geq (10c/p) \log (10/\delta)$ then
after $n=0.625 c/p$ new nodes, with probability $1-\delta$ the
expected length of the path is at least $c/p$.
\end{claim}

The main concern in the analysis has two folds. The first and the
easier case, is how long a phase would last, in expectation, since
the number of waiting nodes is the sum of the original waiting nodes
and the new arriving nodes. By Corollary
\ref{cor:NASP-phase-duration} this would be at most $O(c/p)$. The
second, and more involve, is bounding the expected number of waiting
nodes at the end of the phase.

\begin{theorem}
For $\theta=c/p$, for any time $t$, the expected number of waiting
nodes at time $t$ in NASP is at most $O((c/p) \log (1/\delta))$.
\end{theorem}

\begin{proof}
In case that at the start of the phase we have $|Q_t|\leq 10(c/p)
\log (10/\delta)$, at the end of the phase we have an expected
increase in the number of waiting nodes, which is the expected
duration of the phase minus the length of the path, which is at most
$1.21c/p-c/p=0.21c/p$, using
Corollary~\ref{cor:NASP-phase-duration}.


The proof is similar to the proof of
Theorem~\ref{thm:batch-algorithm}. Let $Q_t$ be the number of
waiting nodes at the end of phase $t$. We would like to show that
the $Q_t$ forms a $(M,K,\rho,\beta)$ random walk. Again, we scale
each $0.625 c/p$ nodes as a ``one unit'', and show the bound for the
$(M,K,\rho,\beta)$ random walk. At the end we multiply by $0.625
c/p$ to get the true bound.

Let $X=0.625 c/p$. By Claim~\ref{claim:long-path} for $|Q_t|\geq
(16X) \log (10/\delta)$ then after $n=X$ new nodes, with probability
$1-\delta$ the expected length of the path is at least $1.6X$. Fix
$\delta=0.1$, after scaling down by $X$, then we have
$(M,K,\rho,\beta)$ random walk with $M=10 \log 100$, $K= 1.6$,
$\rho=0.99$ and $\beta=0.58$.

By Theorem~\ref{thm:rand-walk}, for such a $(M,K,\rho,\beta)$ random
walk, we have that the expected value is at most $M+ O(K)=O(1)$ and
with probability $1-\delta$ it is at most $O(\log(1/\delta))$.
Scaling back by $0.625c/p$ derives the theorem.
\end{proof}

\section{Combined {\tt Greedy-Batch} Algorithm}

In this section we present an algorithm that combines \greedy\ and \batch\ 
in a simple way, and achieves $O(1/p)$ expected waiting time.  While the guarantee
is not better than \batch, the combined algorithm has the appealing property that 
it always makes a match whenever a match from the current end-of-path exists.

The idea of the combined algorithm is to run \batch, maintaining a set $Q$ of ``old'' nodes and a set $V_{fp}$ of ``new'' nodes that have arrived since the start of the current phase.  However, instead of simply waiting for $V_{fp}$ to reach size $c/p$, if a node arrives with an in-edge from the current end-of-path $v^e$, then we run \greedy.   That is, if $V_{fp}$ has size less than $c/p$ and a node arrives that can be matched, we run \greedy; if $V_{fp}$ has size greater than or equal to $c/p$ and a node arrives that can be matched, we run {\tt FAIR-PATH}.

To analyze this algorithm, we use one of the key properties of \greedy, which is that given multiple out-edges to choose from, it always chooses the edge to the node that has been waiting longest, regardless of any structural properties of the nodes.  Therefore, after each run of \greedy, the nodes still in $V_{fp}$ (i.e., not yet matched) remain uniform random, except for the new end of path $v^e$.   Thus, except for an additional expected $1/p$ nodes needed to produce the first out-edge from $v^e$ after $V_{fp}$ has reached size $c/p$, the analysis of each phase of \batch\ is just as before.

The one change to the overall analysis of \batch\ is that the length of each phase is no longer exactly $c/p$ but is rather a random variable.  In particular, in addition to the expected $1/p$ nodes needed to produce the first out-edge from $v^e$ after $V_{fp}$ reaches size $c/p$, there is also the number of new nodes $\Delta$ matched in runs of \greedy\ during the phase.   This is potentially a concern because in the (rare) event that the set $Q$ of old nodes is large, if $\Delta$ is also large in this phase then this increases the average overall queue size.  However, note that by definition of \greedy, if the current node has any edge to an old node, then such an edge will be taken since the old nodes by definition have been waiting longer than the new nodes.  Moreover if $Q$ has size greater than $c'/p$ for sufficiently large $c'$, a new node will have probability at least $0.9$ of having an edge to an old node.  Therefore, if $Q$ is large, then the expected number of old nodes matched by \greedy\ in this phase is at least $0.9\Delta_{large}$, where $\Delta_{large}$ is the number of new nodes matched by \greedy\ in the current phase while $Q$ has size greater than $c'/p$.   Therefore, we can charge matches of new nodes by \greedy\ in intervals where $Q$ is large to progress in decreasing the size of $Q$.  In particular, if $Q$ is large then for every $c/p$ new nodes matched by \greedy\ we make at least as much progress in reducing $Q$ as in a phase of \batch.

%

Similar to Theorem~\ref{thm:batch-algorithm} we have,
\begin{theorem}
%
For $p<0.04$ and $c>10$, the expected waiting time in the {\tt
Greedy-Batch} algorithm with parameter $c$ is $O(c/p)$.
%
%
\end{theorem}

\section{Multiple altruistic donors}

%
%
%
%
%



Recall that we extend our model as follows. Assume we have $R$
multiple donors. Each donor will create a path, so we have $R$
disjoint paths.

 It is
very surprising that having a small number (less than $1/p$) of
multiple donors does not significantly reduce the expected waiting
time. On the other hand, if we have a large number of multiple
donors (more than $(1/p)\log (1/p)$) then the greedy algorithm have
a constant expected waiting time. We remark, that conditioned on the
fact that a node is not matched immediately, the waiting time of the
mode is $\Omega(1/p)$.

\begin{theorem}
For $R\leq 1/p$ the expected waiting time of any algorithm is at
least $\Omega(1/p)$. For $R\ge (\log (1/p))/p$ the expected waiting
time of the greedy algorithm is $O(1)$.
\end{theorem}

\begin{proof}
Consider the case that $R\leq 1/p$. Assume that the number of
waiting nodes $q_t$ is at most $1/p$ (otherwise we are done). When a
new node arrives, with probability at least $1/e$ it does not have
any incoming edge, and therefore it clearly cannot be immediately
added to any of the current $R$ paths.
Conditioned on the fact that the new incoming node does not have any
incoming edges, the expected time until the new node will have some
incoming edge is $1/p$, so the expected waiting time is at least
$1/(ep)$. Since either $q_t\geq 1/p$ or if $q_t\leq 1/p$ then with
probability at least $1/e$ we have $w_t\geq 1/p$, which implies that
$E[q_t]+E[w_t]\geq 1/(ep)$. Similar to
Theorem~\ref{thm:lower-bound-any}, we have that the expected waiting
time is at least $\Omega(1/p)$.

For $R\geq (1/p) \log (1/p)$ we have that the probability that a new
node $v_t$ is not matched immediately at time $t$ is $(1-p)^R\leq
p$. At any future time $\tau>t$, the probability that $v_t$ is added
to some path is at least the probability that $v_\tau$ is
immediately add and there is an edge $v_\tau\rightarrow v_t$, i.e.,
$(1-(1-p)^R)p>(1-p)p$. This implies that the waiting time of $v_t$,
until it is matched is at most $1/((1-p)p)$. Therefore the expected
waiting time is $O(1)$. Equivalently, the expected queue size is
$O(1)$.

We remark that the expected additional waiting for a node which is not matched immediately, is $\Omega(1/p)$.
Therefore, the  $O(1)$  waiting time is mainly due to nodes which are matched immediately as they arrive.
\end{proof}

\section{Future Directions}
In this paper we have focused on a clean sparse random graph model
in which each edge $(u,v)$ is present with some small probability
$p$.  This model can be viewed as a setting in which patients are
all highly sensitized (so $p$ is low) and all bring an O-donor (so
we do not need to consider blood-type incompatibility), and was
studied in the static case in \cite{AGRR12}.

A more complex model would incorporate blood-type incompatibility,
as done in the static dense-graph case in
\cite{AshlagiR11,AshlagiR14}.  One challenge here from the online
perspective is that if the probability of an O-donor is the same as
the probability of an O-patient, then no method can produce bounded
queue sizes.  In particular, if we define the random variable $X$ to
be the number of O-patients seen minus the number of O-donors seen,
or zero if that difference is negative, then the queue size must be
at least $X$, and after $t$ time-steps $E[X] = \Omega(\sqrt{t})$.
So, no matter what algorithm is used for extending the altruistic
donor path, queue sizes will grow with $t$.

However, an interesting model to consider for future work would be
one where patients are ``encouraged'' to bring O-donors, though not
100\% of them do.  An interesting question there would be how large
a fraction of O-donors would be needed to achieve the bounded queue
sizes we get in the model studied here.

\bibliographystyle{plain}
\bibliography{path-bib}

\newpage
\appendix
\section{Long paths in random graphs}
\label{app:long}

In this Appendix we include the results regarding long paths in
random graphs. The following lemma from \cite{Krivelevich16} has the
essence of the methodology of generating long paths using DFS.

\begin{lemma}[\cite{Krivelevich16}]
\label{lemma:long-path}
Given a graph $G(V,E)$ such that for any two
disjoint subsets $S_1,S_2\subset V$ of size at least $k$ there is an
edge, then the DFS will return a path of length at least $|V|-2k$.
\end{lemma}

\begin{lemma}
Given a graph $G(V,E)$ such that between any two disjoint subsets
$S_1,S_2\subset V$ of size at least $k$ there is an edge, assuming
that $|V|\geq 3k$, then $G$ has a connected component of size at
least $|V|-k+1$
\end{lemma}

\begin{proof}
Assume for contradiction that all the connected components of $G$
are of size at most $|V|-k$. Let $C_1, \ldots , C_m$ be the
connected components of $G$. Clearly $m\geq 2$, otherwise we have a
single connected component of size $|V|$. If $m=2$ then
$|C_1|+|C_2|=|V|$, therefore for one connected component, say $C_1$
we have $|C_1|\geq |V|/2\geq k$. If $|C_2|\leq k-1$ we are done,
since $|C_1|=|V|-|C_2|\geq |V|-k-1$. Otherwise $|C_2|\geq k$. This
implies that we have two subsets, $C_1$ and $C_2$, each of size at
least $k$, which do not share an edge. contradiction.

For $m\geq 3$, assume that $|C_1|\geq \cdots \geq |C_m|$. Clearly,
if $|C_1|\geq |V|-k+1$ we are done. If $|C_1| \geq k$ but $|C_1|\leq
|V|-k$, then we have a contradiction by considering $S_1=C_1$ and
$S_2= V-C_1$, since $|S_2|=|V|-|C_1|\geq k$. Otherwise $C_1\leq
k-1$. Consider the index $r$ such that $\sum_{i=1}^r |C_i|\leq k-1$
and $\sum_{i=1}^{r+1} |C_i|\geq k$. Note that this implies that
since all the connected components are of size at most $k-1$, we
have that $\sum_{i=1}^{r+1} |C_i|\leq 2k-2$. Since $|V|\geq 3k$, we
have that $\sum_{i=r+2}^{m} |C_i|\geq k$. This implies that for
$S_1=\bigcup_{i=1}^{r+1} C_i$ and $S_2=\bigcup_{i=r+2}^m C_i$ we
have a contradiction.
\end{proof}

\begin{corollary}
\label{cor:cc}
Given a graph $G(V,E)$ such that between any two disjoint subsets
$S_1,S_2\subset V$ of size at least $k$ there is an edge, assuming
that $|V|\geq 3k$, any set of $S_3$ of at least $k$ nodes has some
$v\in S_3$ which belongs to a connected component of size at least
$|V|-k+1$.
\end{corollary}

\begin{corollary}
\label{cor:start-nodes}
Given a graph $G(V,E)$ such that between any two disjoint subsets
$S_1,S_2\subset V$ of size at least $k$ there is an edge, assuming
that $|V|\geq 3k$, for any set of $S_3$ of at least $k$ nodes has
some $v\in S_3$ which has a path of length $|V|-2k$.
\end{corollary}

\begin{proof}
By Corollary~\ref{cor:cc} there is a node $v\in S_3$ which belongs
to a connected component of size at least $|V|-k+1$. Consider the
DFS from node $v$. At any time while the number of nodes not visited
is at least $k$, the number of nodes from which the DFS backtracked
is at most $k-1$. Otherwise the set of nodes backtracked and the set
of nodes not visited are both at least size $k$ and they are
disjoint, which is a contradiction to the hypothesis in the
corollary.
\end{proof}

\section{Random walks}
\label{app:random-walk}

In this section we outline the proof of
{Theorem~\ref{thm:rand-walk}}.

We define a new random walk $Y_t$ that stochastically dominates
$Q_t$. Similar to $Q_t$ we have $Y_1=0$, and $Y_{t+1}=Y_t+1$ or
$Y_{t+1}=\max\{Y_t-Z_t,M\}$. However, if $Y_t\geq M+1$ then $Z_t =
K$ with probability $\rho$ and otherwise $Z_t=0$. As before, $\rho K
=1+\beta $ where $\beta >0$.

We consider the steady state distribution of the random walk $Y_t$,
where $s_\ell$ is the probability that $Y_t=\ell$ for $\ell\geq M$.
The steady state probability need to satisfy the following
identities.
\begin{align}
\forall \ell\geq M+1\;\;\; s_{\ell+1} & = (1-\rho)s_\ell +\rho
s_{\ell+K+1}\\
s_M &= \rho \sum_{i=1}^K s_{M+i}\\
\forall \ell\;\;\; s_\ell &\geq 0\\
\sum_{\ell} s_\ell & = 1
\end{align}
The first identity implies that to reach a value of $Y_{t+1}=\ell+1$
either $Y_t=\ell$ and $Y_{t+1}=Y_t+1$ or $Y_t=\ell+K+1$ and $Z_t=K$.
The second identity states that to reach $Y_{t+1}=M$ then $Y_t\in
[M+1,M+K]$ and $Z_t=K$. The last two identities simply state that
$s_\ell$ is a distribution.

We will show that there is a solution to the identities such that
for $\ell\geq M+1$ we have $s_\ell=c\alpha^\ell$ for a constant
$c>0$. In such a case the first identity becomes
\[
\forall \ell\geq M+1\;\;\; c\alpha^{\ell+1} = (1-\rho)c\alpha^\ell
+\rho c\alpha^{\ell+K+1}
\]
Simply dividing across by $c\alpha^\ell$ we have
\[
\alpha = (1-\rho) +\rho \alpha^{k+1}
\]
We rename $k+1=k'$ and re-parameterize $\alpha$ using $x>0$ as
\[
\alpha= 1-\frac{x}{k'}
\]
This implies
\[
1-\frac{x}{k'} = 1-\rho +\rho (1-\frac{x}{k'})^{k'}\approx 1-\rho
+\rho e^{-x}
\]
Re-organizing
\[
\rho k' = x + \rho k' e^{-x}
\]
Recall that $\rho k = 1+\beta$, this implies that $\rho k'=
1+(\beta+\rho)$. Let $\beta'=\beta+\rho$. We have
\[
0=e^{-x}-1+\frac{x}{1+\beta'} \triangleq f(x)
\]
Note that $f(0)=0$, but $x=0$ implies $s_\ell=1$ and clearly
violates the fact that it should sum to $1$ (be a distribution).
Also note that $f(1+\beta')>0$ and $f(\epsilon)<0$ for small enough
$\epsilon>0$, so there is another root in $(0,1+\beta]$.

Using the Taylor series expansion we have that for $x\in(0,1)$,
\[
1-x+\frac{x^2}{2}-\frac{x^3}{6} < e^{-x} < 1-x+\frac{x^2}{2}
\]

This implies that
\[
\frac{x}{1+\beta'}-x+\frac{x^2}{2}-\frac{x^3}{6} < f(x) <
\frac{x}{1+\beta'} -x+\frac{x^2}{2}
\]
equivalently,
\[
\frac{-\beta'}{1+\beta'}x+\frac{x^2}{2}-\frac{x^3}{6} < f(x) <
\frac{-\beta'}{1+\beta'}x+\frac{x^2}{2}
\]
For $x=\frac{2\beta'}{1+\beta'}$ we have that the LHS (the upper
bound) is zero. Therefore,
\[
f(\frac{2\beta'}{1+\beta'})<0
\]
For $x=\frac{4\beta'}{1+\beta'}$
is
\[
\frac{-4\beta'^2}{(1+\beta')^2}+\frac{8\beta'^2}{(1+\beta')^2}-\frac{64\beta'^3}{6(1+\beta')^3}
=
\frac{4\beta'^2}{(1+\beta')^2}\left(1-\frac{8\beta'}{3(1+\beta')}\right)>0
\]
where the inequality follows since $\beta'<3/5$. This implies that
\[
f(\frac{4\beta'}{1+\beta'})>0
\]
Therefore, for some $x= \frac{\gamma\beta'}{1+\beta'}$, we have
$f(x)=0$, where $\gamma\in[2,4]$.

We can now consider the second identity and have
\[
s_M = \rho \sum_{i=1}^K s_{M+i} = c\rho \sum_{i=1}^K \alpha^i =
c\rho \alpha \frac{1-\alpha^{k+1}}{1-\alpha} = c\rho (1-\frac{x}{k})
\frac{k}{x}(1-e^{-x}(1-\frac{x}{k}))< c\frac{\rho k }{x}
\]

Clearly we have $s_\ell>0$. We now need to set $c>0$ such that they
sum to $1$.
\[
s_M+\sum_{i=1}^\infty s_{M+i}= s_M + c \sum_{i=1}^\infty \alpha^i
=s_M + c \frac{1}{1-\alpha}= s_M + c \frac{k}{x} <
c\frac{k(1+\rho)}{x}
\]
This implies that $c \in[\frac{x}{k(1+\rho)}, \frac{x}{k}]$. Since
$x\in[\frac{2\beta'}{1+\beta'},\frac{4\beta'}{1+\beta'}]$, we have
that $c \in[\frac{2\beta'}{k(1+\beta')(1+\rho)},
\frac{4\beta'}{(1+\beta')k}]$

\begin{claim}
the expected value of $Y_t$ is at most $M+
\frac{k}{x}<M+\frac{k(1+\beta')}{2\beta'} $.
\end{claim}

\begin{proof}
The claim follows by considering the steady state distribution:
\begin{align*}
E[Y_t]&=M s_M +\sum_{i=1}^\infty s_{M+i}(M+i)\\
& = M + \sum_{i=1}^\infty i s_{M+i}\\
& = M + c\sum_{i=1}^\infty i \alpha^i \\
& = M + \frac{c\alpha}{1-\alpha}\sum_{i=0}^\infty i \alpha^i(1-\alpha) \\
& = M+\frac{c\alpha^2}{(1-\alpha)^2} = M+
c\frac{k^2}{x^2}(1-\frac{x}{k})^2 \\
&< M+c\frac{k^2}{x^2} < M+\frac{k}{x}\\
&\leq M+k\frac{1+\beta'}{2\beta'}
\end{align*}
\end{proof}
This implies that $E[Y_t]<M+O(k)$, and since $Y_t$ dominates $Q_t$
we have that $E[Q_t]  < M+O(K)$. For the high probability we have
the following.

\begin{claim}
With probability $1-\delta$ we have $Q_t\leq M+A$, where
$A=\frac{k(1+\beta')}{2\beta'}\ln \frac{c k}{\delta \lambda}$
\end{claim}

\begin{proof}
The probability of states with more than $A$ are
\[
\sum_{i=A}^\infty s_{M+i} = \sum_{i=A}^\infty c\alpha^i =
\frac{c\alpha^A}{1-\alpha}
\]
Recall that $1-\alpha=1-\frac{x}{k}\geq
1-\frac{2\beta'}{k(1+\beta')}$. Also, $\alpha^A \leq
(1-\frac{2\beta'}{k(1+\beta')}^A$. Using the value of $A$ we have
that $\alpha^A \leq \frac{ck(1+\beta')}{2\beta'\delta}$, and
therefore, the probability is bounded by $\delta$.
\end{proof}

%
%
%
%

\end{document}